\documentclass{llncs}
\usepackage[T1]{fontenc}
\usepackage{tikz}
\usepackage{times} \usepackage{microtype}
\usepackage{mathrsfs}
\usetikzlibrary{patterns}
\usepackage{amsmath, amssymb} 
\usepackage[hidelinks]{hyperref}
\usepackage[htt]{hyphenat}
\usepackage{subcaption}
\usepackage{array}
\usepackage{graphicx}
\usepackage{amssymb}
\usepackage{nimbusmononarrow} 
\usepackage{orcidlink}

\usepackage{color}

\usepackage{ amssymb }

\newcommand{\twiddle}{\mathrel|\joinrel\sim}        

\newcommand{\FK}{\mathbb{K} = (G,M,I)}                      
\newcommand{\underlinesymbol}[1]{\underline{\vphantom{y}#1}}
\newcommand{\K}{\mathbb{K^{\preceq}}}                       
\newcommand{\Kc}{{\mathfrak{B}}(\mathbb{K})}                     
\newcommand{\Cl}{\underlinesymbol{\mathfrak{B}}(\mathbb{K})}
\newcommand{\Tc}{\mathfrak{T}(\K)}                           
\newcommand{\TCl}{\underlinesymbol{\mathfrak{T}}(\K)}
\newcommand{\minO}[1]{\underlinesymbol{#1}'}                
\begin{document}
\title{Non-monotonic Extensions to Formal Concept Analysis via Object Preferences}
\titlerunning{Non-monotonic Extensions to FCA}
%
\author{Lucas Carr\inst{1}\orcidlink{0000-0001-7464-8422} \and Nicholas Leisegang\inst{1}\orcidlink{0000-0002-8436-552X} \and Thomas Meyer\inst{1}\orcidlink{0000-0003-2204-6969} \and Sebastian Rudolph\inst{2,3}\orcidlink{0000-0002-1609-2080} }
\authorrunning{L.Carr et al.}
%
\institute{University of Cape Town and CAIR, South Africa \\ 
\email{\{crrluc003,lsgnic001\}@myuct.ac.za, tommie.meyer@uct.ac.za} 
\and Technische Universität Dresden, Germany \\ 
\and
Center for Scalable Data Analytics and Artificial Intelligence Dresden/Leipzig\\
\email{sebastian.rudolph@tu-dresden.de}}
\maketitle              
\begin{abstract}
    Formal Concept Analysis (FCA) is an approach to creating a conceptual hierarchy in which a \textit{concept lattice} is generated from a \textit{formal context}. That is, a triple consisting of a set of objects, $G$, a set of attributes, $M$, and an incidence relation $I$ on $G \times M$. A \textit{concept} is then modelled as a pair consisting of a set of objects (the \textit{extent}), and a set of shared attributes (the \textit{intent}).  Implications in FCA describe how one set of attributes follows from another. The semantics of these implications closely resemble that of logical consequence in classical logic. In that sense, it describes a monotonic conditional. The contributions of this paper are two-fold. First, we introduce a non-monotonic conditional between sets of attributes, which assumes a preference over the set of objects. We show that this conditional gives rise to a consequence relation that is consistent with the postulates for non-monotonicty proposed by Kraus, Lehmann, and Magidor (commonly referred to as the KLM postulates). We argue that our contribution establishes a strong characterisation of non-monotonicity in FCA. To our knowledge, this is a novel view of FCA as a formalism which supports non-monotonic reasoning. We then extend the influence of KLM in FCA by introducing the notion of \textit{typical} concepts through a restriction placed on what constitutes an acceptable preference over the objects. Typical concepts represent concepts where the intent aligns with expectations from the extent, allowing for an exception-tolerant view of concepts. To this end, we show that the set of all typical concepts is a meet semi-lattice of the original concept lattice. This notion of typical concepts is a further introduction of KLM-style typicality into FCA, and is foundational towards developing an algebraic structure representing a concept lattice of prototypical concepts. 

\end{abstract}
\section{Introduction}
\label{Introduction}
Formal Concept Analysis (FCA) is a lattice-theoretic approach to representing and reasoning about concepts and hierarchies in data. The view of concepts adopted in FCA has clear philosophical underpinnings, describing a concept as a pair representing the dualism between extension, \textit{what a concept refers to}, and intension, \textit{what a concept means} \cite{ganter2012foundations,ganter2016conceptual,R2006}. As such, the setting of FCA works with data in the form of a formal context, describing a relationship between objects and attributes.

An important component of FCA involves discovering implications in the data which describe a complete correspondence between sets of attributes \cite{R2006}. A concern might be that complete correspondence is too strict a requirement, and implications may fail to capture relationships that, while useful, are only partial correspondences between attribute sets. To illustrate this point, we introduce a variation of the customary example of penguins, where we look at species of animals and some of their attributes (see \autoref{Example-running} for a more detailed description) \cite{Kaliski_2020,kraus_nonmonotonic_2002}. In our example we have three objects which have the attribute \texttt{bird}: duck, robin, and penguin. Of these, duck and robin also have the attribute \texttt{flies}, and obviously penguin does not. Consequently, the implication \texttt{bird} $\rightarrow$ \texttt{flies} is not valid in our context. One may argue that this behaviour is completely correct. Clearly, not all birds fly. A compelling response is that birds that do not fly are exceptions to the rule, and that the ability to express that birds usually do fly is useful. 

Association rules offer an existing approach to address this concern. These rules introduce a notion of confidence, which quantifies the proportion of data that conforms to a given rule, if the confidence is above a user-defined threshold are then accepted \cite{ganter2016conceptual}. While association rules are capable of capturing partial relationships, this approach may be considered somewhat blunt, relying on something analogous to \textit{majority rules}. The reliance on a threshold for rule acceptance limits the expressivity of these rules.

We propose an alternate approach which adopts the view that we can define a preference relation over the objects in our data. While we are -- at present -- agnostic about the origin of this preference relation, from an applied perspective it might represent an external sentiment that one object is more \textit{typical} than another. Partial implications between attribute sets can then be captured through the idea that the correspondence holds for \textit{preferred} objects \cite{Kaliski_2020,kraus_nonmonotonic_2002}.

Once we have a preference relation over the set of objects, and if we regard the relation as representative of the typicality of objects, we can introduce a notion of a \textit{typical concept}. The motivation for typical concepts, and what they mean, rests on the notion that, given some set of attributes, we want a coherent way to derive a concept which has the attributes we \textit{typically} expect. Continuing with the example from before, we might expect that, while the intent of the ``bird concept'' derived from the attribute \texttt{bird} does not include the attribute \texttt{flies}, it may be useful to have a notion of a typical concept which includes the attributes of prototypical birds.

The remainder of the paper is structured as follows: \autoref{Section_FCA} provides an account of some basic notions in FCA. Then, \autoref{Logical Consequence and Nonmonotonic Reasoning} provides an introduction to classical notions of consequence, non-monotonic reasoning, and an intuition for how these translate to FCA. We introduce the announced extensions to FCA in \autoref{FCA-NMR} and discuss the semantics and properties of the consequence relation they allow. In \autoref{Typ-Conc} we develop the notion of typical concepts, derived from the preferential view of objects. In \autoref{Related Work} we discuss some attempts to bring KLM-style defeasibility into other logical systems. We also discuss existing work within FCA which proposes to introduce more expressive notions of concepts. 

\section{Formal Concept Analysis}
\label{Section_FCA}
Two fundamental notions in FCA are \textit{formal context} and \textit{formal concept}. A formal context has a set-theoretic definition, where three sets of objects, attributes, and a binary relation are defined. However, in reasonably sized instances it can be represented as a cross-table, where rows represent the objects, and columns are the attributes. Naturally, the presence of an \textquoteleft$\times$\textquoteright \space at a position in the table indicates that the respective object has the attribute from the respective column.
\begin{definition}
    \label{def:formal-context}
    A \emph{formal context} is a triple $(G,M,I)$ where $G$ refers to a finite, non-empty set of objects, $M$ to a finite, non-empty set of attributes, and $I \subseteq G\times M$ is an incidence relation. For $(g,m)\in I$, we say object $g$ has attribute $m$, otherwise $g$ does not have attribute $m$.
\end{definition}
\begin{example}
    \label{Example-running}
    \autoref{fig:example-formal-context} shows a cross-table representation of a formal context, describing different animal species and some of their attributes. 
    \begin{figure}[h]
        \vspace{-2ex} 
        \centerline{
            \begin{tabular}{r@{\hspace{0.2cm}}|ccccc}
                                 & \texttt{ northern} & \texttt{ southern} & \texttt{ flies} & \texttt{ antarctic} & \texttt{ bird} \\ \hline
                \texttt{orca}    & $\times$           & $\times$           &                 & $\times$            &                \\[0.2em]
                \texttt{duck}    & $\times$           & $\times$           & $\times$        &                     & $\times$       \\[0.2em]
                \texttt{robin}   & $\times$           & $\times$           & $\times$        &                     & $\times$       \\[0.2em]
                \texttt{penguin} &                    & $\times$           &                 & $\times$            & $\times$       \\
            \end{tabular}
        }
        \caption{A context of animal species and some of their attributes. For example, Orcas are found in the northern hemisphere, southern hemisphere, and the antarctic.}
        \label{fig:example-formal-context}
    \end{figure}
    \vspace{-0.75cm}
\end{example}
In the build-up to defining formal concepts, one introduces two operators which define an order-reversing Galois connection between the power sets $\mathcal{P}(G)$ and $\mathcal{P}(M)$. These two operators, which share the same notation, describe the derivation from a set of objects to the shared attributes, and vice versa \cite{ganter2016conceptual,ganter2012foundations}.
\begin{definition}
    \label{definition: derivation operators}
    In a formal context, $\FK$, \emph{the derivation operator} $(\cdot)'$ is defined for sets $A\subseteq G$ and $B \subseteq M$ as:
    \begin{align*}
        A' & := \{m \in M \mid \forall g \in A, (g,m) \in I\} \\
        B' & := \{g \in G \mid \forall m \in B, (g,m) \in I\}
    \end{align*}
    For a set $A \subseteq G $ of objects (resp. attributes), $A'$ is just a set of attributes (resp. objects), and so $A''$ would be a set of objects (resp. attributes). The double application of derivation operators is in fact a closure operator, meaning it is \emph{extensive, idempotent, and monotonic}.
\end{definition}
Concepts are represented as a pairs of sets of objects and attributes; as one might expect, these sets are not arbitrary. In fact, they completely determine one another. Conventionally, we might denote a concept as $(A,B)$ with $A\subseteq G, B\subseteq M$. This may be a useful notation to express the demarcation between the extent and intent, but it is not strictly necessary, for, $(A, A')$ and $(B', B)$ would refer to the same concept.
\begin{definition}
    \label{def:formal-concept}
    Let $\mathbb{K} = (G,M,I)$ be a formal context. A pair $(A,B)$, where $A\subseteq G$ and $B \subseteq M$, is a \emph{formal concept} iff $A' = B$ and $B' = A$.
\end{definition}
In fact, given an arbitrary set of objects, $A \subseteq G$, one can entirely describe a concept by $(A'', A')$ -- the same principle holds for sets of attributes \cite{ganter2012foundations}. If we refer back to \autoref{Example-running}, starting from two sets of objects and attributes, respectively, $\{\texttt{duck, penguin}\}$ and $\{\texttt{antarctic}\}$, we find the following concepts:
{\small
\begin{align*}
    \big(\{\texttt{duck,penguin}\}'' , \{\texttt{duck,penguin}\}' \big) & = \big(\{\texttt{duck,robin,penguin}\}, \{\texttt{southern,bird}\}\big) \\ 
    \big(\{\texttt{antarctic}\}' , \{\texttt{antarctic}\}'' \big) & = \big(\{\texttt{penguin,orca}\}, \{\texttt{southern,antarctic}\}\big)
\end{align*}}
These two concepts suggest that one can introduce a partial ordering $\leq$ to concepts (induced by the subset relationship on their extents), establishing the notion of sub- and super-concepts.
Equipped with this ordering, the set of all concepts of a formal context $\mathbb{K}$, denoted by $\Kc$, forms a complete lattice $\Cl := (\Kc, \leq)$, referred to as a concept lattice \cite{ganter2016conceptual,ganter2012foundations}. \autoref{fig:example_lattice} provides an example of a concept lattice. Attributes connected to a concept from above (inclusive) are part of that concept's intent, while objects connected from below (inclusive) are in the extent. For a set of concepts, $\{(A_t, B_t) \mid t \in T\} \subseteq \mathfrak{B}({\mathbb{K})}$, the supremum (\textit{smallest super-concept}), and infimum (\textit{greatest sub-concept}) are given by
\begin{align*}
    \bigvee_{t\in T} (A_t,B_t)    & := \bigl( (\bigcup_{t\in T} A_t)'',  \bigcap_{t\in T} B_t\bigr)  \ \ \mbox{as well as}\\
    \bigwedge_{t \in T} (A_t,B_t) & := \bigl( \bigcap_{t\in T} A_t,  (\bigcup_{t\in T} B_t)'' \bigr) . \\
    \vspace{-0.5cm}
\end{align*}
\subsection{Implications}
\label{Preliminaries - FCA - Implications}
A significant aspect of FCA involves implications between sets of attributes. These express the notion that certain attributes indicate the presence of others, in all objects of a formal context \cite{ganter2016conceptual,R2006}.
\begin{definition}
    \label{def:classical implications}
    Let $\mathbb{K} = (G,M,I)$ be a formal context. An \emph{implication over $M$} is an expression of the form $A \rightarrow B$ with $A,B \subseteq M$. We say that the implication is \textit{respected} by another set $C\subseteq M$, iff $A \not \subseteq C$ or $B \subseteq C$. $\mathbb{K}$ respects the implication, written $\mathbb{K} \models A\rightarrow B$, iff for every $g \in G$, $g'$ respects $A \rightarrow B$.
\end{definition}
It is useful to note that a formal context $\mathbb{K} = (G,M,I)$ respects an implication $A\rightarrow B$ over $M$ exactly if $B\subseteq A''$, which, in turn, is equivalent to $A' \subseteq B'$.

\section{Logical Consequence and Nonmonotonic Reasoning}
\label{Logical Consequence and Nonmonotonic Reasoning}
\subsection{Logical Consequence}
\label{logical consequence}
A classical notion of logical consequence describes the circumstances under which one sentence is said to follow (logically) from another \cite{Etchemendy_Tarski}. In a model-theoretic view, it is given as follows. 
\begin{definition}
    \label{definition: logical consequence}
    For two sentences $\alpha, \beta$ in the language $\mathcal{L}$, we say that $\beta$ is a \emph{logical consequence} of $\alpha$, expressed as $\alpha \vDash \beta$, iff for every valuation $u\in \mathcal{U}$ where $u \Vdash \alpha$ then also $u \Vdash \beta$.
\end{definition}
This definition can easily be generalised to the notion of \textit{entailment}: given a set of facts it would be useful to know \textit{what else} we can know, or, what is entailed \cite{Kaliski_2020}.
\begin{definition}
    \label{definition: classical entailment}
    Given a set of sentences $\mathcal{KB}$, another sentence $\alpha$ is \emph{entailed} by $\mathcal{KB}$ iff for every valuation $u \in \mathcal{U}$ with $u \Vdash \mathcal{KB}$ also $u \Vdash \alpha$ holds. This is expressed as $\mathcal{KB} \models \alpha$.
\end{definition}
A consequence operator, $\mathcal{C}n$, provides a general way to derive all sentences that should follow from a set of sentences $\mathcal{K}\mathcal{B}$, under \textit{some} notion of logical consequence \cite{Kaliski_2020}. Using classical entailment from \autoref{definition: classical entailment}, we could have $\mathcal{C}n(\mathcal{K}\mathcal{B}) := \{\alpha \mid \mathcal{K}\mathcal{B} \models \alpha \}$ \cite{citkin2021consequencerelationsintroductiontarskilindenbaum,Kaliski_2020}. 
A Tarskian consequence operator satisfies the following properties (which describe a closure operator). 
\begin{align}
     & \text{Monotonocity:} & \; \textit{if } \; \Gamma \subseteq \Gamma'  \; \text{then} \; \mathcal{C}n(\Gamma) \subseteq \mathcal{C}n(\Gamma' ) \\
     & \text{Idempotence:}  & \; \mathcal{C}n(\Gamma) = \mathcal{C}n(\mathcal{C}n(\Gamma))                                                         \\
     & \text{Inclusion:}    & \; \Gamma \subseteq \mathcal{C}n(\Gamma)
\end{align}
\subsection{Consequence Relations}
\label{Consequence Relations}
A more abstract notion of consequence is a \textit{consequence relation}. This is a set of pairs, $\{(\Gamma_1, \gamma_1), \ldots, (\Gamma_n, \gamma_n), \ldots \}$, where it is typical to allow $\Gamma_i$ to represent a set of formulae, and $\gamma_i$ to represent a single formula in the language. The inclusion of a pair $(\Gamma_i, \gamma_i)$ in the consequence relation—denoted as $\Gamma_i \vdash \gamma_i$—means that $\gamma_i$ can be inferred from $\Gamma_i$ \cite{gabbay1995general,Kaliski_2020}.

Consequence relations may be characterised by the properties they satisfy and, as such, correspond to a certain kind of reasoning---an example of this is found in \autoref{definition: rational consequence} \cite{gabbay1995general}. Conversely, an algorithmic description of a kind of reasoning may in turn give rise to a consequence relation, the properties of which provide strong intuition for the pattern of reasoning. 


\subsection{Non-monotonic Reasoning}
\label{Non-monotonic Reasoning}

Non-monotonic reasoning is concerned with developing formal reasoning process in which a conclusion drawn under a premise can be withdrawn under the addition of another premise \cite{kraus_nonmonotonic_2002}. The justifications for why we may want to reason non-monotonically are easy to accept -- it is quite obvious that when humans reason we do so under the implicit assumption that, upon receiving new information we can change our mind. Without this assumption, navigating life would be very difficult; we would hesitate to come to any conclusion out of fear it may be the wrong one. Moreover, we frequently make statements for which we know exist exceptions, under the assumption that upon encountering an exception it would be treated as such \cite{makinson2005go}.

In the setting of propositional logic, the problem is often introduced in the following way \cite{Kaliski_2020,kraus_nonmonotonic_2002,stanford-nonmonotonic}: we want to accept that "penguins are birds", "birds usually fly" and "penguins do not fly". However, when we translate this into classical logic with the propositions $\texttt{penguin}\rightarrow \texttt{bird}$, $\texttt{bird}\rightarrow \texttt{fly}$ and $\texttt{penguin}\rightarrow \neg \texttt{fly}$, we are forced to conclude $\neg\texttt{penguin}$. With respect to the attribute logic of FCA, we find a similar issue. Consider the context in \autoref{Example-running}. We do not have $\{\texttt{bird}\} \rightarrow \{\texttt{flies}\}$ as a valid implication, since penguins are an object with \texttt{bird} but not \texttt{flies}. This illustrates the final point made in the last paragraph -- that monotonicity prevents us from making statements with known exceptions. 

In the realm of concepts, we encounter another instance of this problem: the concept determined by \texttt{bird} is not a sub-concept of the concept determined by $\texttt{flies}$. Again, we lack the expressivity for the idea that \textit{typical} birds are flying animals. In \autoref{FCA-NMR} and \autoref{Typ-Conc} we propose a solution to this lack of expressivity drawing inspiration from \cite{kraus1990nonmonotonic,kraus_nonmonotonic_2002}.       
\subsection{Rational Consequence Relations}
\label{Rational Consequence Relations}
The style of non-monotonic reasoning we aim to develop in FCA is that of rational consequence relations. A rational consequence relation, $\twiddle_R$, is based on preferential logics, where an order (or ranking) over valuations conveys the notion that certain valuations are preferred to others \cite{kraus1990nonmonotonic,lehmann1994what,Shoham}. The semantics describe a notion of consequence where, if $\alpha, \beta$ are formulae in the language then $\beta$ is a rational consequence of $\alpha$, $\alpha \twiddle_R \beta$, iff $\beta$ is true in all the \textit{most preferred} models of $\alpha$. 

The KLM postulates, introduced in \cite{kraus1990nonmonotonic,lehmann1994what}, with the inclusion of \textit{Rational Monotonicity} are a set of inference rules that characterise a rational consequence relation.

\begin{definition}
    \label{definition: rational consequence}
    A consequence relation, $\twiddle$, constitutes a rational consequence relation iff it satisfies Reflexivity, Left Logical Equivalence (LLE), Right Weakening (RW), Cut, And, Or, Cautious Monotonicity and Rational Monotonicity. 
\end{definition}
\noindent
\textbf{Reflexivity} \quad $A \twiddle A$ \quad Reflexivity is a somewhat basic notion of any notion of consequence; it essentially prevents a self-defeating pattern of reasoning where given something, you conclude as a consequence not that thing. \\

\noindent
\textbf{Left Logical Equivalence (LLE)} \quad $\frac{ \models A \equiv B, A \twiddle C}{B \twiddle C}$ \quad LLE enforces the notion that two things that are equivalent (under a coherent notion of equivalence) should have the exact same consequences.\\

\noindent
\textbf{Right Weakening (RW)} \quad $\frac{ \models A \rightarrow B, C\twiddle A}{C \twiddle B}$ \quad RW states that the consequence of a classical implication, which has a defeasible consequence as a premise, can itself be derived from the original defeasible implication.\\

\noindent
\textbf{Cut} \quad $\frac{A \wedge B \twiddle C, A \twiddle B}{A \twiddle C}$ \quad Cut allows us to use existing defeasible consequences in the premise of a new defeasible implication. However, the new implication is subject to failure should the defeasible conclusion in its premise be retracted.\\

\noindent
\textbf{Or} \quad $\frac{A \twiddle C, B \twiddle C}{A \vee B \twiddle C}$ \quad Or is that given distinct premises with a common defeasible consequence, we should be able to draw this conclusion from the disjunction of the premises -- that is, we need not explicitly know which one is true.\\

\noindent
\textbf{And} \quad $\frac{A \twiddle B, A \twiddle C}{A \twiddle B \wedge C}$ \quad And tells us that two consequences can be concluded at the same time.\\

\noindent
\textbf{Cautious Monotonicity (CM)} \quad $\frac{A \twiddle B, A \twiddle C}{A \wedge B \twiddle C}$ \quad CM ensures that adding a defeasible consequence that could already have been derived from our premises should never invalidate another defeasible consequence that could be derived from our original premises. \\

\noindent
\textbf{Rational Monotonicity (RM)} \quad $\frac{A \twiddle B, A \not \twiddle \neg C}{A \wedge C \twiddle B}$ \quad RM expresses the notion, similar to conflicts in Default Logic, that only when a premise that was expected to be false is added, should we retract a rational consequence. In another sense, we can assume that when new information that does not explicitly contradict existing knowledge, we can retain existing conclusions \cite{Kaliski_2020}. 


\section{Introducing Nonmonotonicity in FCA}
\label{FCA-NMR}
\subsection{Extended Formal Context}
\label{FCA-NMR_Ext-Cont}

As a precursor to the non-monotonic conditionals in \autoref{FCA-NMR_NMR-Impl},  we introduce the notion of an extended formal context.
\begin{definition}
    \label{definition-extended_context}
    Let $\mathbb{K} = (G,M,I)$ be a formal context. An \emph{extended formal context} simply adds a partial order over $G$. We denote this as a quadruple, $\K = (G,M,I,\preceq)$.
\end{definition}
The partial ordering is intended to convey a preference relation between objects. That is, for two objects, $g, h \in G$, if $g \preceq h$, then we say that $g$ is more preferred, or more typical, in comparison to $h$. For example, $\texttt{robin}\preceq \texttt{penguin}$ tells us a robin is ``more typical'' than a penguin in our context.
\subsection{Minimisation}
\label{FCA-NMR_Min}
The preference relation between objects provided by an extended formal context enables us to define how the derivation operators might behave when restricted to only care about minimal objects. We formalise this notion as a \textit{minimised derivation}.
\begin{definition}
    \label{definition: minimised derivation}
    Let $\K = (G,M,I,\preceq)$ be an extended formal context, and let $B\subseteq M$. Then the \emph{minimised derivation}, $\minO{B}$, of $B$, is the set
    \[\minO{B} := \{g \in B' \mid \nexists h \in B' \text{ such that } h \prec g\}\]
\end{definition}
As a reminder, for a set of attributes, $B \subseteq M$, $B'$ provides us with all the objects which \textit{have} all attributes of $B$. Following this, $B''$ extends $B$ by including all other attributes which are shared by those objects which share attributes $B$. 

If $\minO{B}$ describes the process of going from a set of attributes to the minimal objects which share $B$, a reasonable question might concern the procedure for returning to a ``closed'' set of attributes: presumably, the set of attributes common to those objects described by $\minO{B}$. This is, in fact, just the composition of \autoref{definition: minimised derivation} and \autoref{definition: derivation operators}, taking the form $(\minO{B})'$ - henceforth referred to as a \textit{minimised-return operation}.
\begin{theorem}
    \label{properties minimised-return operator}
    Let $\K = (G, M, I, \preceq)$ be an extended formal context, the minimised-return operator, $(\minO{\cdot})'$, applied to a set of attributes is nonmonotonic, extensive, and idempotent.
\end{theorem}
\begin{proof}
    If the minimised-return operator were monotonic, for the sets $A,B \subseteq M$, $A\subseteq B$ would imply $(\minO{A})' \subseteq (\minO{B})'$. Refer back to the formal context in \autoref{Example-running} extended with the partial order \texttt{robin} $\preceq$ \texttt{penguin}. Let $A := \{\texttt{bird}\}$ and $B := \{\texttt{bird,antarctic}\}$. Clearly, $A \subseteq B$, however, application of the minimised-return operator gives $(\minO{A})' = \{\texttt{bird, flies, southern, northern}\}$ and $(\minO{B})' = \{\texttt{bird,antarctic,southern}\}$, and so $(\minO{A})' \not \subseteq (\minO{B})'$.

    To show that the minimised-return operator is extensive, it should hold that for a set $A \subseteq M$, $A \subseteq (\minO{A})'$. We already know that in the classical case, $A \subseteq A''$. So, for all objects $g \in A'$, it is the case that $A \subseteq g'$. It is also clear that the minimised derivation of a set of attributes is a subset of the classical derivation -- that is, $\minO{A} \subseteq A'$. This means that for all objects $h \in \minO{A}$, $A \subseteq h'$. It should then be straightforward to see that $A$ is a subset of the set given by $\bigcap \{g' \mid g\in \minO{A}\}$, which is equivalent to $A \subseteq (\minO{A})'$. 

    The final task is to show that the minimised-return operator is idempotent. Given a set $A\subseteq M$, let $C := (\minO{A})'$. It is enough to show that $\minO{A} = \minO{C}$. $C$ is the set of attributes such that all objects, $g \in \minO{A}$, have $C$ in their intent. It is also clear, since all objects in $\minO{A}$ have $A$ in their intent, that $A\subseteq C$. Then, for $h \in \minO{C}$ it follows that $h$ has attributes $C$, and thus $A$, in its intent. Since $h$ is minimal w.r.t. $C$, and $g \in C'$, it cannot be the case that $g \preceq h$. Conversely, since $g$ is minimal w.r.t. $A$, and $h \in A'$, it cannot be that $h \preceq g$. Then, since both $g$ and $h$ are elements of $A'$ and $C'$, and are arbitrary, the sets $\minO{A}$ and $\minO{C}$ refer to the same objects. Then, applying a derivation to the same sets would of course yield the same result, and so the minimised-return operator is idempotent. \qed
\end{proof}
\subsection{Non-monotonic Conditional through Minimisation}
\label{FCA-NMR_NMR-Impl}
From the discussion in \autoref{Preliminaries - FCA - Implications} it is clear that the semantics of classical implications in FCA can be defined through applications of closure operators on sets of attributes. With the introduction of the minimised-return operator, one might expect to find a different notion of consequence. 

The minimised-return operator still refers to the same kind of process as the closure-operator, going from a set of attributes to a set of objects, and then back to a set of attributes. The similarity enable us to define a semantics which is only a minor departure from the original implication. For now we will use \textit{non-monotonic conditional} to refer to the minimised-derivation based implications, and denote this by $\rightsquigarrow$.  
\begin{definition}
    \label{def non-monotonic conditional}
    Let $\K = (G,M,I,\preceq)$ be an extended formal context, with $A,B\subseteq M$. $\K$ respects a \emph{non-monotonic conditional} $\K \models A\rightsquigarrow B$ iff $B \subseteq (\minO{A})'$ which is equivalent to $\minO{A} \subseteq B'$.
\end{definition}
Although we have yet to prove the characterisation, the non-monotonic conditional from above describes a notion of consequence with strong semblance to the preferential consequence relations discussed in \autoref{Rational Consequence Relations}. 
\begin{example}
    \label{example: nixon diamon variation}
    Let the cross table below represent an extended formal context, where the partial order over objects is: \texttt{Achilles $\preceq$ Jason, Minos}.
    \begin{small}
    \begin{figure}
        \begin{center}
            \begin{tabular}{r@{\hspace{0.2cm}}|ccc}
                                  & \texttt{ warrior} & \texttt{ demigod} & \texttt{ hero} \\ \hline
                \texttt{Jason}    & $\times$          &                   & $\times$       \\[0.2em]
                \texttt{Achilles} & $\times$          & $\times$          & $\times$       \\[0.2em]
                \texttt{Minos}    &                   & $\times$          &
            \end{tabular}
        \end{center}
        \caption{Extended formal context of figures in Greek mythology}
    \end{figure}
    \end{small}

    In Greek mythology there is a strong correspondence between figures who are heroes, and demigods. So much so that there is debate as to whether \textit{hero} and \textit{demigod} actually have distinct meanings. Thus, with some creative licence, we may wish to express the notion that \texttt{hero} \textit{usually implies} \texttt{demigod}. The classical implication, of course, fails to express this - as \texttt{Jason} is a counter-example. Instead, we express this notion as the non-monotonic conditional \texttt{hero} $\rightsquigarrow$ \texttt{demigod}.
\end{example}

In our definition of classical implications, see \autoref{def:classical implications}, we first introduced $\rightarrow$ as an object-level operator, saying that a set $C\subseteq M$ respects an implication $A \rightarrow B$ iff $A \not \subseteq C$ or $B \subseteq C$ (we can think of $C$ as an object intent). However, most of the time when we talk about implications in FCA, we really refer to the meta level notion analogous to entailment that $K\models A\rightarrow B$.

For non-monotonic conditionals it makes even less sense to speak about the object level, we are interested in what our ordering allows us to conclude -- and the ordering is explicitly on the meta-level. Consequently, whenever we speak about a non-monotonic conditional, we are speaking about the implication over the entire formal context.

\subsection{Non-monotonic Conditionals and the KLM Postulates}
\label{non-monotonic conditionals and the KLM Postulates}
We now present an argument that the consequence relation given by the non-monotonic conditional described in the previous sub-section satisfies all the KLM postulates required to be a characterisation of rational consequence relations. Before we prove this characterisation in \autoref{Appendix}, we should remind ourselves that the postulates are usually described in the language of some truth-theoretic logic. A consequence being that some initial work needs to be translated into the attribute-logic formalism considered here. 

For \textit{Reflexivity} and \textit{RW}, the translation is obvious and doesn't require additional intuition. All that is required for \textit{LLE} is to describe what a notion of equivalence between sets of attributes means. To this end, we say that in a formal context two attribute sets are equivalent, $A\equiv B$, iff $A' = B'$. \textit{Cut} says that if $\alpha \wedge \beta \twiddle \gamma$ and $\alpha \twiddle \beta$ then $\alpha \twiddle \gamma$. The conjunction of two formulae in a truth-theoretic logic requires satisfaction of each formula. In attribute logic, the equivalent notion is given by the union of two attribute sets, satisfied by objects with \textit{all} attributes from both sets. The conjunction in \textit{CM} and \textit{And} is re-phrased in the same way. 

It is not immediately clear what \textit{RM} might mean, as it is uncommon to talk about implications with a negation. We use the following definitions to make this notion explicit. 
\begin{definition}
    \label{def: negation in implication}
    Let $\FK$ be a formal context and $A \rightarrow \neg B$ an implication over $M$. We say that a set $C \subseteq M$ respects the implication $A \rightarrow \neg B$ iff $A \not\subseteq C$ or $B \not\subseteq C$. $\mathbb{K}$ respects the implication, $\mathbb{K} \models A\rightarrow \neg B$, iff for every object $g$, $g'$ respects $A\rightarrow \neg B$. Equivalently, $A' \cap B' = \emptyset$.
\end{definition}

\begin{definition}
    \label{def: negation in preferential implcation}
    An extended formal context, $\K$, respects the non-monotonic conditional $A\rightsquigarrow \neg B$ iff, for every object $g$ in $\minO{A}$, $B \not \subseteq g'$. Equivalently, $\K \models A \rightsquigarrow \neg B$ iff $\minO{A} \cap B' = \emptyset$.
\end{definition}

The intuition is that one set of attributes corresponds to the absence of another if there are no objects which have both sets in their intent. On the object level of a formal context, for an object $g \in G$, if $g \not \models A\rightarrow B$ then $g \models A\rightarrow \neg B$. We do not, however, have that if $\mathbb{K} \not \models A\rightarrow B$ then $\mathbb{K} \models A \rightarrow \neg B$. Translating the intuition of $A\rightarrow \neg B$ to $A \rightsquigarrow \neg B$ is straightforward. The non-monotonic conditional, $A\rightsquigarrow \neg B$, concerns only objects $g\in \minO{A}$, which, by definition of the minimised derivation operator, must have $A$ in their intent. Consequently, we need only show that $g$ does not have $B$ in its intent. 

We have yet to address the \textit{Or} postulate. The difficulty here is that attribute logic does not have (nor is it intuitive to introduce) a notion of disjunction \cite{5368560}. As such, we recognise a departure from the KLM framework, and omit the \textit{Or} postulate entirely. Consequence relations that satisfy only \textit{Reflexivity, LLE, RW, Cut}, and \textit{CM} are considered \textit{cumulative consequence relations} \cite{kraus_nonmonotonic_2002}. However, while we do not have\textit{Or}, we show that we still have \textit{Rational Monotonicity}, and continue to regard our relation as a characterisation of \textit{rational consequence*} (we indicate absence of the \textit{Or} postulate with a star). 
\begin{theorem}
    \label{Satisfies rational consequence}
    The consequence relation derived from $\rightsquigarrow$ satisfies Reflexivity, Left Logical Equivalence, Right Weakening, Cut, And, Cautious Monotonicity, and Rational Monotonicity, and thus, constitutes a rational consequence relation*. 
\end{theorem}
This gives us a notion of non-monotonic consequence that corresponds to what we may think \textit{logical} non-monotonic reasoning should look like, as described by rational consequence. We are able to use intermediary results as the basis for more conclusions through the use of \textit{Cut} and \textit{CM} \cite{Kaliski_2020,stanford-nonmonotonic}. \textit{RM} is a useful property in the sense that it reduces the work one has to do - unless we have some consequence that is contradictory to what we expect, we can retain existing conclusions. The notion of consequence also preserves classical implications in FCA, and so it is strictly more expressive than regular implications in FCA \cite{stanford-nonmonotonic}. 

Now is a good place to remind ourselves what we gain from introducing non-monotonic conditionals to FCA. If we accept rational consequence relations as a good account of non-monotonic reasoning, and that we can impose a preference relation on objects in a formal context, non-monotonic conditionals provide FCA with the expressivity to reason about what may be the \textit{typical} case.  
\section{Typical Concepts}
\label{Typ-Conc}
The partial ordering over objects was introduced as a means to developing a non-monotonic consequence relation. We now show how this ordering, with some refinement, gives rise to a notion of typical concepts.
\subsection{Naive Notions of a Typical Concept}
\label{Typ-Conc_Intro}
\begin{definition}
    \label{typical concept 1}
    Let $\K = (G,M,I,\preceq)$ be an extended formal context. Then, for a set $A \subseteq M$, we define a \emph{typical concept} as a concept of the form \[\left(\left(\minO{A}\right)'', \left(\minO{A}\right)'\right).\] The set of all typical concepts of an extended formal context is denoted $\Tc$.
\end{definition}
In this case, the intent of a typical concept is the set of attributes common to the minimal objects which have $A$ in their intent. The extent is then all objects which have this extended set of attributes. Referring back to \autoref{Example-running}, we add the partial order \texttt{duck} $\preceq$ \texttt{penguin} and \texttt{robin} $\preceq$ \texttt{penguin}, expressing the sentiment that ducks and robins are more \textit{typical} than penguins, but incomparable to orcas. Given only the attribute $\texttt{bird}$, $(\{\texttt{duck,robin}\}, \{\texttt{northern,southern,flies,bird}\})$ is the derived typical concept. Without this notion of a typical concept, the derived concept would be $(\{\texttt{duck,robin,penguin}\},\{\texttt{southern,bird}\})$. Of course, both of these concepts exist in the classical concept lattice; however, typical concepts provide an instrument to, given only the condition of \texttt{bird}, to arrive at a concept which we might consider a more natural characterisation of things that are birds. We should remind ourselves that, while this example is picked to match real-world expectations, we should think of ``expectations'' as being an expression of the preference relation on objects. 

Perhaps un-intuitively, this formulation allows for a typical concept to contain non-typical objects in its extent, as long as they share all the attributes of their typical counterparts. We should recognise, however, that they do not contribute to describing the intent, they just happen to be consistent with it. Although we ascribe specific meaning to $\minO{A}$, it can be considered to be an arbitrary set of objects. Then, per \autoref{def:formal-concept}, an arbitrary typical concept, $((\minO{X})'', (\minO{X})')$ where $X\subseteq M$, is always equivalent to some formal concept. As such, we can define a map from the set of concepts to the set of typical concepts.
\begin{definition}
    \label{map}
    Let $\K = (G,M,I,\preceq)$ be an extended formal context,then $\mathfrak{B}(\K)$ denote the set of concepts of $\K$, and $\Tc$ denote the set of typical concepts of $\K$. We define a map:
    \[\varphi : \Kc \mapsto \Tc\]
    For a concept $(A,B) \in \mathfrak{B}(\K)$, $\varphi$ is defined as:
    \[\varphi(A,B) := ((\minO{B})'', (\minO{B})')\]
\end{definition}
The mapping function inherits idempotency from the minimised-return operator. For a typical concept $(A,B) \in \Tc$, there must exist some set $C\subseteq M$ such that $(\minO{C})' = B$. It follows that $(\minO{(\minO{C})'})' = (\minO{C})' = B$. Consequently, $\varphi(A,B) = (A,B)$, or, typical concepts map to themselves. Differently put, $\varphi(A,B) = \varphi(\varphi(A,B))$. {We also note that by the extensive property of minimized-derivation, for any concept $(A,B)$, we have $(\minO{B})'\supseteq B$ and therefore $\varphi(A,B)=((\minO{B})'', (\minO{B})')\leq (A,B)$. That is, $\varphi$ is \textit{anti-extensive}.
\begin{theorem}
    \label{verify}
    In an extended formal context, $\K = (G,M,I,\preceq)$, let $(A,B) \in \mathfrak{B}({\K})$, then  $(A,B) \in \Tc$ iff $(\minO{B})'' = A$.
\end{theorem}
\begin{proof}
    Let $\K = (G,M,I,\preceq)$ be an extended formal context with a concept, $(A,B) \in \mathfrak{B}  ({\K})$. Suppose $(A,B)$ were a typical concept, then there exists a set $C\subseteq M$ such that $A = (\minO{C})''$ and $B = (\minO{C})'$, per the definition of a typical concept. Through the properties of the minimised-return operator in \autoref{properties minimised-return operator} we have that $\minO{C} = \minO{(\minO{C})'}$, equivalently, $\minO{C} = \minO{B}$. Then, obviously, $(\minO{C})'' = (\minO{B})''$. \qed
\end{proof}
A next point of departure might be to investigate the structure the set of typical concepts. Certainly, we have that $\Tc \subseteq \mathfrak{B}(\K)$. Then in particular, we want to investigate whether the subset of ``typical'' objects forms a sub-lattice, if we can form a sub-lattice of only those ``typical'' concepts, we are able to use all the same algebraic tools to analyze the data in our context ``prototypically'', as opposed to the classical analysis which relies on strict implications and concepts. Using the same approach to ordering typical concepts as we did for formal concepts (see \autoref{Section_FCA}), two examples make it clear that the current formulation unfortunately does not guarantee that the pair $(\Tc, \leq)$, henceforth denoted $\TCl$, will form any type of lattice.
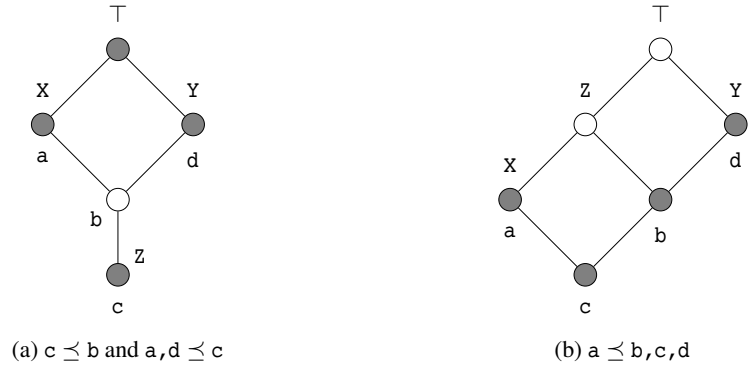
\begin{figure}[t]
    \centering
    \begin{subfigure}[b]{0.45\textwidth}
        \centering
        \begin{tikzpicture}[
                concept/.style={circle, draw, line width=0.3pt, minimum size=0.3cm},
                label_above/.style={font=\footnotesize, above=0.1cm},
                label_below/.style={font=\footnotesize, below=0.1cm},
                line/.style={draw, line width=0.3pt},
                greyed/.style={text opacity=0.5}
            ]
            \node[concept, fill=gray] (top) at (0,3) {};
            \node[label_above] at (top.north) {$\top$};
            \node[concept, fill=gray] (two_1) at (-1,2) {};
            \node[label_above] at (two_1.north) {\texttt{X}};
            \node[label_below] at (two_1.south) {\texttt{a}};
            \node[concept, fill=gray] (two_2) at (1,2) {};
            \node[label_above] at (two_2.north) {\texttt{Y}};
            \node[label_below] at (two_2.south) {\texttt{d}};
            \node[concept] (three_1) at (0,1) {};
            \node[label_below,left=0.1cm] at (three_1.south) {\texttt{b}};
            \node[concept, fill=gray] (four_2) at (0,0) {};
            \node[label_above,right=0.1cm] at (four_2.north) {\texttt{Z}};
            \node[label_below] at (four_2.south) {\texttt{c}};
            \draw[line] (top) -- (two_1);
            \draw[line] (top) -- (two_2) -- (three_1);
            \draw[line] (two_1) -- (three_1) -- (four_2);
        \end{tikzpicture}
        \caption{\texttt{c} $\preceq$ \texttt{b} and \texttt{a,d} $\preceq$ \texttt{c}}
        \label{fig:lattice1}
    \end{subfigure}
    \hfill
    \begin{subfigure}[b]{0.45\textwidth}
        \centering
        \begin{tikzpicture}[
                concept/.style={circle, draw, line width=0.3pt, minimum size=0.3cm},
                label_above/.style={font=\footnotesize, above=0.1cm},
                label_below/.style={font=\footnotesize, below=0.1cm},
                line/.style={draw, line width=0.3pt},
                greyed/.style={text opacity=0.5}
            ]
            \node[concept] (top) at (0,3) {};
            \node[label_above] at (top.north) {$\top$};
            \node[concept] (two_1) at (-1,2) {};
            \node[label_above] at (two_1.north) {\texttt{Z}};
            \node[concept,fill=gray] (two_2) at (1,2) {};
            \node[label_above] at (two_2.north) {\texttt{Y}};
            \node[label_below] at (two_2.south) {\texttt{d}};
            \node[concept, fill=gray] (three_1) at (-2,1) {};
            \node[label_above] at (three_1.north) {\texttt{X}};
            \node[label_below] at (three_1.south) {\texttt{a}};
            \node[concept, fill=gray] (three_2) at (0,1) {};
            \node[label_below] at (three_2.south) {\texttt{b}};
            \node[concept, fill=gray] (four_2) at (-1,0) {};
            \node[label_above] at (four_2.north) {\texttt{}};
            \node[label_below] at (four_2.south) {\texttt{c}};
            \draw[line] (top) -- (two_1) -- (three_1) -- (four_2);
            \draw[line] (top) -- (two_2) -- (three_2);
            \draw[line] (two_1) -- (three_2) -- (four_2);
        \end{tikzpicture}
        \caption{\texttt{a} $\preceq$ \texttt{b,c,d}}
        \label{fig:lattice2}
    \end{subfigure}
    \caption{Classical concept lattices, with the typical concepts derived from the respective ordering marked as grey}
    \label{fig:concept-lattices}
    \vspace{-0.5cm}
\end{figure}
\autoref{fig:lattice1} is the concept lattice of some formal context, $\mathbb{K}$. If one were to extend $\mathbb{K}$ with the partial order \texttt{c} $\preceq$ \texttt{b} and \texttt{a,d} $\preceq$ \texttt{c}, the resultant set of typical concepts would be the four concepts marked with grey nodes. We pay special attention to the concepts
\begin{align*}
    C_1 := \Big((\minO{X})'', (\minO{X})'\Big) & = \Big((\texttt{a})'', (\texttt{a})'\Big)   = \Big(\{\texttt{a,b,c}\}, \{\texttt{X}\}\Big) \\
    C_2 := \Big((\minO{Y})'', (\minO{Y})'\Big) & = \Big((\texttt{d})'', (\texttt{d})'\Big)   = \Big(\{\texttt{b,c,d}\}, \{\texttt{Y}\}\Big)
\end{align*}
The greatest common sub-concept of these two concepts is given by
\begin{align*}
    C_1 \land C_2 & = \big(\{\texttt{a,b,c}\} \cap \{\texttt{b,c,d}\}, (\{\texttt{X}\} \cup \{\texttt{Y}\})''\}\big) = \big( \{\texttt{b,c}\}, \{\texttt{X,Y}\} \big)
\end{align*}
However, from \autoref{verify}, it is clear that this is not a typical concept, as $(\minO{\texttt{\{X,Y\}}})'' = \{\texttt{c}\} \not = \{\texttt{b,c}\}$. This result tells us that the subset of typical concepts is not closed under meets. If we shift attention to \autoref{fig:lattice2}, it becomes clear that it is not closed under joins either, nor does it guarantee a lattice with respect to the concept-lattice ordering. For the typical concepts
\begin{align*}
    C_1 := \Big((\minO{X})'', (\minO{X})'\Big)             & = \Big((\texttt{a})'', (\texttt{a})'\Big)   = \Big(\{\texttt{a,b}\}, \{\texttt{X}\}\Big)       \\
    C_2 := \Big((\minO{\{Y,Z\}})'', (\minO{\{Y,Z\}})'\Big) & = \Big((\texttt{b,c})'', (\texttt{b,c})'\Big)   = \Big(\{\texttt{b,c}\}, \{\texttt{Y,Z}\}\Big)
\end{align*}
the least common super-concept is given by
\begin{align*}
    C_1 \lor C_2 & = \big((\{\texttt{a,c}\} \cup \{\texttt{b,c}\})'', \{\texttt{X,Z}\} \cap \{\texttt{Y,Z}\}\}\big) = \big( \{\texttt{b,c}\}, \{\texttt{Z}\} \big).
\end{align*}
%
For $C_1 \lor C_2$, the closure of the minimised-derivation of the intent, $(\minO{\texttt{Z}})'' = \{\texttt{a,c}\}$, does not equal the concept extent. So, we do not have that typical concepts are closed under joins. To make matters worse, the top element of the lattice is the concept $(\{\texttt{a,b,c,d}\}, \emptyset)$, but not a typical concept, as $(\minO{\emptyset})'' = (\texttt{a})'' = \{\texttt{a,c}\}$, and so we do not have a lattice structure of typical concepts. In fact, we do not even have an upper bound to the set of typical concepts in \autoref{fig:lattice2}.
\subsection{Restriction on the Partial Order}
\label{Typ-Conc_Restr}
One cause of the aforementioned problem lies in the fact that we allow non-minimal objects to be included in typical concepts. This allows for rankings where non-minimal objects, which should be in the meet or join of two concepts, are \textit{lost}. As a solution to this, we may restrict the partial orders to prevent non-minimal objects from being included in typical concepts.
\begin{definition}\label{def:valid-order}
    Let $\K = (G,M,I,\preceq)$ be an extended formal context. A \emph{valid partial order} over objects is one which for all $A \subseteq M$, $\minO{A} = (\minO{A})''$.
\end{definition}
Consider a set of attributes $X\subseteq M$ and then take $Y\subseteq M$ to be the set of all attributes common to the minimal objects which satisfy $X$ -- of course, $Y$ would be the set $X$ with additional attributes. The restriction to orders $\preceq$ that are valid ensures that any object that has all the attributes in $Y$ must itself be minimal w.r.t. $X$. Put differently, $Y$ serves as a total characterisation of the minimal objects satisfying attributes $X$, if an object matches this characterisation, then it should have been one of these minimal objects.
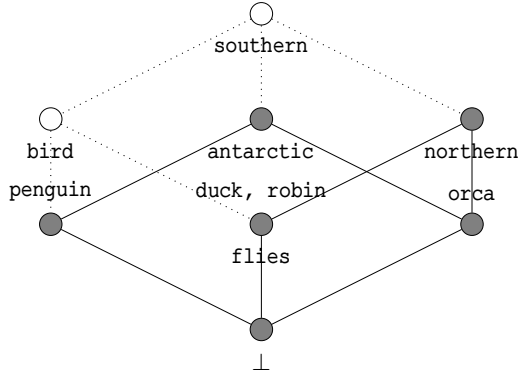
\begin{figure}[t]
    \centering
    \begin{tikzpicture}[
            scale=1, 
            concept/.style={circle, draw, line width=0.3pt, minimum size=0.2cm},
            label_above/.style={font=\small, above=0.05cm},
            label_below/.style={font=\small, below=0.05cm},
            line/.style={draw, line width=0.3pt},
            greyed/.style={text opacity=0.5}
        ]
        \node[concept, greyed] (top) at (0.7,2.1) {};
        \node[label_below, greyed] at (top.south) {\texttt{southern}};
        \node[concept, greyed] (two_1) at (-2.1,0.7) {};
        \node[label_below, greyed] at (two_1.south) {\texttt{bird}};
        \node[concept,fill=gray] (two_2) at (0.7,0.7) {};
        \node[label_below] at (two_2.south)  {\texttt{antarctic}};
        \node[concept,fill=gray] (two_3) at (3.5,0.7) {};
        \node[label_below] at (two_3.south)  {\texttt{northern}};
        \node[concept,fill=gray] (three_1) at (-2.1,-0.7) {};
        \node[label_above] at (three_1.north) {\texttt{penguin}};
        \node[concept,fill=gray] (three_3) at (3.5,-0.7) {};
        \node[label_above] at (three_3.north) {\texttt{orca}};
        \node[concept,fill=gray] (four_1) at (0.7,-0.7) {};
        \node[label_below] at (four_1.south) {\texttt{flies}};
        \node[label_above] at (four_1.north) {\texttt{duck, robin}};
        \node[concept,fill=gray] (bottom) at (0.7,-2.1) {};
        \node[label_below] at (bottom.south) {$\bot$};
        \draw[dotted] (top) -- (two_1);
        \draw[dotted] (top) -- (two_2);
        \draw[dotted] (top) -- (two_3);
        \draw[dotted] (two_1) -- (three_1);
        \draw[line] (three_1) -- (bottom);
        \draw[dotted] (two_1) -- (four_1);
        \draw[line] (two_2) -- (three_1);
        \draw[line] (two_2) -- (three_3);
        \draw[line] (two_3) -- (four_1) -- (bottom);
        \draw[line] (two_3) -- (three_3) -- (bottom);
    \end{tikzpicture}
    \caption{The concept lattice for \autoref{Example-running}. Given the preference order \texttt{duck} $\preceq \texttt{penguin}$ and \texttt{robin} $\preceq$ \texttt{penguin}, the dark grey concepts show the structure of typical concepts}
    \label{fig:example_lattice}
    \vspace{-0.5cm}
\end{figure}
If we accept this restriction on the partial order, the set of typical concepts preserves meets, that is, the greatest common sub-concept of two typical concepts is itself a typical concept. Before this is shown, \autoref{typical concept 1}, \autoref{map}, and \autoref{verify} might be amended to reflect that the closure of a set $\minO{A}$, where $A\subseteq M$, is superfluous.
\begin{theorem}
    For an extended formal context, $\K = (G,M,I,\preceq)$, with a restriction on the partial-order such that for any $A\subseteq M$, $\minO{A} = (\minO{A})''$, the greatest common sub-concept of two typical concepts is itself a typical concept. Hence, the subset of typical concepts form a $\wedge$-subsemilattice of the concept lattice.
\end{theorem}
\begin{proof}
        We note that it is a sufficient condition to show that $\minO{A}\cap \minO{B}=\minO{C}$ for some  $C\subseteq M$. We claim that we can choose $C=(\minO{A})'\cup (\minO{B})'$. Note
        \begin{align*}
            \minO{A}\cap \minO{B} & = (\minO{A})''\cap (\minO{B})'' =((\minO{A})'\cup (\minO{B})')',
        \end{align*}
        where the first equality follows from our assumed condition, and the second from our preliminary lemma for Galois connections. Then $\underline{\minO{A}\cap \minO{B}}=\underline{((\minO{A})'\cup (\minO{B})')'}$. Note that
        \[ \underline{\minO{A}\cap \minO{B}}=\{g\in \minO{A}\cap \minO{B}\mid \nexists h\in\minO{A}\cap \minO{B}\text{ such that }h<g\}.\]
        Therefore, if $g\in \minO{A}\cap \minO{B}$, since $g$ is a minimal member of $A'$ and of $B'$ there is no $h$ in $A'$ or in $B'$ with $h< g$ by definition and so $g\in \underline{\minO{A}\cap \minO{B}}$. Therefore $\minO{A}\cap \minO{B}\subseteq \underline{\minO{A}\cap \minO{B}}$, and clearly $\underline{\minO{A}\cap \minO{B}}\subseteq \minO{A}\cap \minO{B}$. If the intersection has no elements $\minO{A}\cap \minO{B}=\emptyset=\underline{\emptyset}$. Hence  $\minO{A}\cap \minO{B}=\underline{\minO{A}\cap \minO{B}}=\underline{((\minO{A})'\cup (\minO{B})')'}$ and we are done. \qed
\end{proof}
It is also worth noting here that the bottom element of the concept lattice is always a typical concept since $\varphi(\bot)\leq \bot$, by anti-extensivity, and so $\varphi(\bot)=\bot$. This result allows us to begin characterizing the algebraic structure of  those concepts considered ``typical'' in an extended formal context. However, there is still much structure that this restriction on the order does not account for. In particular, we do not have that $\TCl$ preserves least common super-concepts -- \autoref{fig:lattice2} is an example of a concept lattice whose order satisfies the condition in \autoref{def:valid-order}, but does not have a $\top$-concept. That being said, we are still able to preserve some structure using a somewhat intuitive restriction on the order in our extended formal context. 

\section{Related Work}
\label{Related Work}
There is a large corpus of work on developing KLM-style notions of non-monotonic reasoning outside of the standard of propositional logic. \cite{casini2022klm}, and more extensively \cite{paterson2022klm}, introduce KLM-style defeasible reasoning to datalog, which can be considered a fragment of first-order logic used for database queries. \cite{britz2019klm,britz2020principles} are recent efforts to introduce a notion of defeasible subsumption to description logics (DL) that follows the KLM properties. Remaining in the realm of DL, \cite{moodley2014practical} is a general characterisation of to KLM in DL $\mathcal{ALC}$. \cite{doi:10.1080/11663081.2017.1397325} investigates the KLM framework in the context of defeasible modalities, introducing new modal operators for \textit{defeasible necessity} and \textit{distinct possibility}. 

With respect to work that extends our view of concepts in FCA, \cite{kent1996rough} introduces \textit{Rough Concept Analysis}, a merging of rough set theory and FCA that uses equivalence classes on objects to define upper and lower-bound approximations of concepts. \cite{yao2016rough} investigates an expansion which enables ``rough concepts'' to be defined not only by objects.      
\section{Conclusions and Future Work}
\label{Concl-Fut}
By extending a formal context with a preference relation on the objects we have introduced a non-monotonic variant of implication between attribute sets which characterise a rational consequence relation*. This strictly increases the expressivity of the attribute logic of FCA by creating a notion of non-monotonicity that corresponds to the KLM view of how a \textit{logical} non-monotonic system should behave. In terms of FCA, we have introduced a way of discovering and representing relationships between attribute sets that tolerates exceptions, and is capable of representing what our data (formal context) shows in the \textit{typical} case. 

With a slight restriction on what constitutes a \textit{valid} preference relation over objects, we presented a formalisation of typical concepts which has its foundations in a KLM-style typicality. We were able to show that with this notion we could create a structure of typical concepts that at least preserved the sub-concept relation, and as such is a \textit{meet-subsemilattice} of the original concept lattice. 

To our knowledge, the introduction of KLM-style typicality, and preferences over a formal context as a whole, presents a novel view on, and non-trivial extension to FCA. 

This work represents the initial investigation of introducing KLM-style typicality into FCA. As such, we believe there is a considerable amount that remains to be looked at. Addressing concerns of existing work, we aim to find an approach to defining typical concepts which is closed under joins (super-concepts), essentially meaning that the structure of the set of typical concepts would be a sub-lattice of the original concept lattice. To this end, we could neatly ``reduce'' a concept lattice to its typical counterpart. In another branch of this work, we would like to investigate the relationship between the set of all non-monotonic conditionals and the typical concept lattice; determining if something analogous to a \textit{canonical basis} can be found.   
\begin{credits}
\subsubsection{\ackname} 
This work is a culmination of effort that took place when students at the University of Cape Town visited the Technical University of Dresden. We would like to express immense appreciation and gratitude to the School of Embedded Composite Artificial Intelligence (SECAI) -- project 57616814 funded by BMBF (the Bundesministerium für Bildung und Forschung) and DAAD (German Academic Exchange Service) -- for funding the visit, and to its members in Dresden who were so welcoming and engaging.

\subsubsection{\discintname}The authors have no competing interests to declare that are
relevant to the content of this article.
\end{credits}
%

\begin{thebibliography}{10}
\providecommand{\url}[1]{\texttt{#1}}
\providecommand{\urlprefix}{URL }
\providecommand{\doi}[1]{https://doi.org/#1}

\bibitem{britz2020principles}
Britz, K., Casini, G., Meyer, T., Moodley, K., Sattler, U., Varzinczak, I.: Principles of klm-style defeasible description logics. ACM Transactions on Computational Logic (TOCL)  \textbf{22}(1),  1--46 (2020)

\bibitem{britz2019klm}
Britz, K., Casini, G., Meyer, T., Varzinczak, I.: A klm perspective on defeasible reasoning for description logics. In: Description Logic, Theory Combination, and All That: Essays Dedicated to Franz Baader on the Occasion of His 60th Birthday, pp. 147--173. Springer (2019)

\bibitem{doi:10.1080/11663081.2017.1397325}
Britz, K., Varzinczak, I.: From klm-style conditionals to defeasible modalities, and back. Journal of Applied Non-Classical Logics  \textbf{28}(1),  92--121 (2018). \doi{10.1080/11663081.2017.1397325}, \url{https://doi.org/10.1080/11663081.2017.1397325}

\bibitem{casini2022klm}
Casini, G., Meyer, T., Paterson-Jones, G., Varzinczak, I.: Klm-style defeasibility for restricted first-order logic. In: International Joint Conference on Rules and Reasoning. pp. 81--94. Springer (2022)

\bibitem{citkin2021consequencerelationsintroductiontarskilindenbaum}
Citkin, A., Muravitsky, A.: Consequence relations an introduction to the tarski-lindenbaum method (2021), \url{https://arxiv.org/abs/2106.10966}

\bibitem{Etchemendy_Tarski}
Etchemendy, J.: Tarski on truth and logical consequence. The Journal of Symbolic Logic  \textbf{53}(1),  51--79 (1988), \url{http://www.jstor.org/stable/2274427}

\bibitem{gabbay1995general}
Gabbay, D.M.: A general theory of structured consequence relations. Theoria: An International Journal for Theory, History and Foundations of Science pp. 49--78 (1995)

\bibitem{ganter2016conceptual}
Ganter, B., Obiedkov, S., Rudolph, S., Stumme, G.: Conceptual exploration. Springer (2016)

\bibitem{ganter2012foundations}
Ganter, B., Wille, R.: Formal concept analysis: mathematical foundations. Springer Science \& Business Media (2012)

\bibitem{Kaliski_2020}
Kaliski, A.: An overview of KLM-style defeasible entailment. Master's thesis, University of Cape Town (2020)

\bibitem{kent1996rough}
Kent, R.E.: Rough concept analysis: a synthesis of rough sets and formal concept analysis. Fundamenta informaticae  \textbf{27}(2-3),  169--181 (1996)

\bibitem{kraus1990nonmonotonic}
Kraus, S., Lehmann, D., Magidor, M.: Nonmonotonic reasoning, preferential models and cumulative logics. Artificial intelligence  \textbf{44}(1-2),  167--207 (1990)

\bibitem{kraus_nonmonotonic_2002}
Kraus, S., Lehmann, D., Magidor, M.: Nonmonotonic {Reasoning}, {Preferential} {Models} and {Cumulative} {Logics} (Feb 2002), \url{http://arxiv.org/abs/cs/0202021}, arXiv:cs/0202021

\bibitem{lehmann1994what}
Lehmann, D., Magidor, M.: What does a conditional knowledge base entail? Artificial Intelligence  \textbf{68}(2) (1994)

\bibitem{5368560}
Lei, Y., Sui, Y., Cao, C.: The logical operations on relations, scales and concept lattices. In: 2009 Fifth International Conference on Semantics, Knowledge and Grid. pp. 160--167 (2009). \doi{10.1109/SKG.2009.37}

\bibitem{makinson2005go}
Makinson, D.: How to go nonmonotonic. In: Handbook of Philosophical Logic, pp. 175--278. Springer (2005)

\bibitem{moodley2014practical}
Moodley, K., Meyer, T., Sattler, U.: Practical defeasible reasoning for description logics. In: STAIRS 2014, pp. 191--200. IOS Press (2014)

\bibitem{paterson2022klm}
Paterson-Jones, G.: KLM-Style Defeasible Reasoning for Datalog. Master's thesis, University of Cape Town (2022)

\bibitem{R2006}
Rudolph, S.: Relational Exploration - Combining Description Logics and Formal Concept Analysis for Knowledge Specification. Universit{\"{a}}tsverlag Karlsruhe (December 2006)

\bibitem{Shoham}
Shoham, Y.: A semantical approach to nonmonotonic logics, p. 227–250. Morgan Kaufmann Publishers Inc., San Francisco, CA, USA (1987)

\bibitem{stanford-nonmonotonic}
Strasser, C., Antonelli, G.A.: {Non-monotonic Logic}. In: Zalta, E.N., Nodelman, U. (eds.) The {Stanford} Encyclopedia of Philosophy. Metaphysics Research Lab, Stanford University, {S}ummer 2024 edn. (2024)

\bibitem{yao2016rough}
Yao, Y.: Rough-set concept analysis: interpreting rs-definable concepts based on ideas from formal concept analysis. Information Sciences  \textbf{346},  442--462 (2016)

\end{thebibliography}

\clearpage
\appendix 
\label{Appendix}
\section*{Appendix}
In the following we present the proofs for \autoref{Satisfies rational consequence} along with their re-characterisation in attribute logic. In case the re-characterisation is non-trivial, we provide some explanation. \autoref{lemma 1} is a consequence of Galois connections, represented in the attribute logic of FCA. \autoref{lemma-imp} is essential for \textit{Cut} and \textit{CM}.  \\
\begin{align*}
    \quad A \rightsquigarrow A  & & \textbf{Reflexivity} &
\end{align*}
\begin{proof}[Reflexivity]
    Let $\K = (G,M,I,\preceq)$ be an extended formal context with $A\subseteq M$. In order to show that $\K \models A \rightsquigarrow A$ it needs to be shown that $\minO{A} \subseteq A'$. We have this by definition of the minimised derivation, that $\minO{A} = \{ g \in A' \mid \nexists h \in A' \text{ such that } h \preceq g\}$. It is then obvious that $\minO{A} \subseteq A'$ and so $\K \models A\rightsquigarrow A$.
\end{proof}
\begin{align*}
    \quad \frac{ \models A \equiv B, A \rightsquigarrow C}{B \twiddle C}  & & \textbf{Left Logical Equivalence (LLE)} &
\end{align*}
\noindent We do not have an existing notion of equivalence between attribute sets; however, for two sets $A,B\subseteq M$, we say that $A\equiv B$ iff $A = B$.   
\begin{proof}[Left Logical Equivalence]
    Let $\K = (G,M,I,\preceq)$ be an extended formal context, where $A,B,C\subseteq M$, $\K \models A \rightsquigarrow C$ and $A = B$. By assumption we have that $\minO{A} \subseteq C'$ and that $\minO{A} = \minO{B}$. Clearly, we then have that $\minO{B} \subseteq C'$ which is equivalent to $\K \models B\rightsquigarrow C$.
\end{proof}
\begin{align*}
    \quad \frac{ \models A \rightarrow B, C\rightsquigarrow A}{C \rightsquigarrow B}  & & \textbf{Right Weakening (RW)} &
\end{align*}
\begin{proof}[Right Weakening]
    Let $\K = (G,M,I,\preceq)$ be an extended formal context, where $A,B,C\subseteq M$, $\K \models A \rightarrow B$ and $\K \models C \rightsquigarrow A$. The classical implication, $A\rightarrow B$, is equivalent to $A' \subseteq B'$. Furthermore, $C\rightsquigarrow A$ is equivalent to $\minO{C} \subseteq A'$. Through transitivity, we have that $\minO{C} \subseteq A' \subseteq B'$, and so $\minO{C} \subseteq B'$. And so, we have that $\K \models C \rightsquigarrow B$
\end{proof}
\begin{lemma}
    \label{lemma 1}
    For any $A,B \subseteq M$, it is the case that  $A' \cap B' = (A\cup B)'$
\end{lemma}
\begin{proof}[Lemma 1]
    For a formal context $\FK$ and two sets $A,B \subseteq M$ we will show the $\subseteq$ in both directions. To begin, let $g$ be an object in $A' \cap B'$. Then, for an arbitrary $m \in A$, it holds that $(g,m) \in I$. Obviously, the same holds for attributes in $B$. So, for an arbitrary $m \in A \cup B$, it must also be the case that $(g,m) \in I$. Given that $g$ and $m$ were arbitrary, this is equivalent to the definition of $(A\cup B)'$ - $\{g \in G \mid \forall m\in A\cup B, (g,m) \in I\}$. Consequently, $A' \cap B' \subseteq (A\cup B)'$.

    To show the other direction is simpler. Let $g$ be an arbitrary object in $(A \cup B)'$, this means that for an arbitrary $m \in (A\cup B)$, $(g,m) \in I$. Trivially, $A \cap B \subseteq (A\cup B)$, so, from before, $(g,m) \in I$ for any $m \in (A\cap B)$. Again, this is our definition for $A' \cap B'$ - $\{g\in G \mid \forall m \in A \cap B, (g,m)\in I \}$ - and we have that $(A \cup B)' \subseteq A' \cap B'$.
\end{proof}
\begin{lemma}
    \label{lemma-imp}
    For some extended formal context $\K = (G,M,I,\preceq)$ where $A,B\subseteq M$, let $\K \models A\rightsquigarrow B$. Then, $\minO{A} = \minO{(A \cup B)}$.
\end{lemma}
\begin{proof}[Lemma 2]
    Let $\K = (G,M,I,\preceq)$ be some extended formal context such that $K\models A \rightsquigarrow B$, with $A,B\subseteq M$. To begin, let $g\in \minO{A}$. We are given that $\minO{A} \subseteq B'$ and so $g\in B'$. From Reflexivity we have that $\minO{A} \subseteq A'$ and so $g\in A'$. This gives us $g\in A' \cap B'$, which, by \autoref{lemma 1}, means $g\in (A\cup B)'$. Now, assume that $g\not \in \minO{(A\cup B)}$. Since $g\in (A\cup B)'$, there must exist some $h \in \minO{(A\cup B)}$ such that $h\preceq g$. But, by Galois connections and the minimised derivation, it is the case that $\minO{(A\cup B)} \subseteq (A\cup B)' \subseteq A'$ (trivially, $A \subseteq A\cup B$). If we remind ourselves that $g\in \minO{A}$ it becomes apparent that $h \preceq g$ is a contradiction. Since $g$ was arbitrary in $\minO{A}$ we then have $\minO{A} \subseteq \minO{(A\cup B)}$.

    To show the other direction, let $j \in \minO{(A\cup B)}$ and assume $j \not \in \minO{A}$. From above we know that $j\in A'$, and so there must exist some $k \in \minO{A}$ with $k \preceq j$. But we have just shown that $\minO{A} \subseteq \minO{(A\cup B)}$. As such $k \in \minO{(A\cup B)}, j\in \minO{(A\cup B)}$ and $k \preceq j$, which is a contradiction. Once again, $j$ is arbitrary, so we have $\minO{(A\cup B)} \subseteq \minO{A}$.
\end{proof}
\begin{align*}
    \quad \frac{A \cup B \rightsquigarrow C, A \rightsquigarrow B}{A \rightsquigarrow C}  & & \textbf{Cut} &
\end{align*}
In \autoref{Logical Consequence and Nonmonotonic Reasoning} we write $A \cup B$ as a conjunction between two formulae. In propositional logic, this refers to valuations where both formulae are satisfied. Moving to an attribute logic, the equivalent notion is that we refer to those objects which have \textit{all} attributes from $A$ and $B$. 
\begin{proof}[Cut]
    Let $\K = (G,M,I,\preceq)$ be an extended formal context, with sets $A,B,C\subseteq M$ and $\K \models A\cup B \rightsquigarrow C$, $A\rightsquigarrow B$. By assumption, we have $\minO{(A \cup B)} \subseteq C'$ and $\minO{A} \subseteq B'$.\autoref{lemma-imp} states that $\minO{A} = \minO{(A\cup B)}$. Then, $\minO{A} \subseteq C'$ which is equivalent to $A\rightsquigarrow C$.
\end{proof}
\begin{align*}
    \quad \frac{A \rightsquigarrow B, A\rightsquigarrow C}{A \rightsquigarrow B \cup C} & & \textbf{And}
\end{align*}
The conjunction of two formulae in a truth-theoretic logic requires satisfaction of each formula. In attribute logic, the equivalent notion is given by the union of two attribute sets, satisfied by objects with all attributes from both sets.
\begin{proof}[And]
    Let $\K = (G,M,I,\preceq)$ be an extended formal context, with sets $A,B,C\subseteq M$ such that $\K \models A\rightsquigarrow B, A\rightsquigarrow C$. By assumption we know both $\minO{A} \subseteq B'$ and $\minO{A} \subseteq C'$. So, $\minO{A} \subseteq B' \cap C'$. Which, by \autoref{lemma 1}, gives us $\minO{A} \subseteq (B\cup C)'$. This is equivalent to $A \rightsquigarrow B \cup C$
\end{proof}

\begin{align*}
   \quad \frac{A \rightsquigarrow B, A \rightsquigarrow C}{A \cup B \rightsquigarrow C} & & \textbf{Cautious Monotonicity (CM)} &
\end{align*}

\begin{proof}[Cautious Monotonicity]
    Let $\K = (G,M,I,\preceq)$ be an extended formal context, with sets $A,B,C\subseteq M$ such that $\K \models A\rightsquigarrow B, A\rightsquigarrow C$. From our assumptions we have, $\minO{A} \subseteq B'$ and $\minO{A} \subseteq C'$. Then, \autoref{lemma-imp} would give us that $\minO{(A\cup B)} \subseteq C'$, which is equivalent to $A\cup B \rightsquigarrow C$.
\end{proof}

\begin{align*}
    \quad \frac{A \rightsquigarrow B, A \not \rightsquigarrow \neg C}{A \cup C \rightsquigarrow B} & & \textbf{Rational Monotonicity (RM)}
\end{align*}

\begin{proof}[Rational Monotonicity]
    Let $\K = (G,M,I,\preceq)$ be an extended formal context, with sets $A,B,C\subseteq M$ such that $\K \models A\rightsquigarrow B, A \not \rightsquigarrow \neg C$. We have that $\minO{A} \subseteq B'$, since $A\rightsquigarrow B$. Additionally, $\minO{A} \cap C' \not = \emptyset$, from $A \not \rightsquigarrow \neg C$ and \autoref{def: negation in preferential implcation}. A consequence of this is that there exists an object, $g\in \minO{A}$, such that $g$ is an element of $C'$ - the objects with $C$ in their intent. Since $\minO{A}\subseteq A'$, $g\in C'$, and \autoref{lemma 1} we have that $g\in (A \cup B)'$. But also by \autoref{lemma 1}, any object in $(A\cup B)'$ would also be in $A'$, and since $g$ is minimal in $A'$, there cannot be an object in $(A\cup B)'$ that is minimal to $g$. And so, if $A\not \rightsquigarrow \neg C$, then $\minO{A} \cap C'$ is equivalent to $\minO{(A\cup B)}$. Then, since $\minO{A} \cap C' \subseteq \minO{A} \subseteq B'$, we have that $\minO{(A\cup C)}\subseteq B'$ and finally, $A\cup C \rightsquigarrow B$.
\end{proof}
\end{document}